\documentclass[pdflatex,sn-mathphys-num]{sn-jnl}

\usepackage{graphicx}%
\usepackage{multirow}%
\usepackage{amsmath,amssymb,amsfonts}%
\usepackage{amsthm}%
\usepackage{mathrsfs}%
\usepackage[title]{appendix}%
\usepackage{xcolor}%
\usepackage{textcomp}%
\usepackage{manyfoot}%
\usepackage{booktabs}%
\usepackage{algorithm}%
\usepackage{algorithmicx}%
\usepackage{algpseudocode}%
\usepackage{listings}%

\theoremstyle{thmstyleone}%
\newtheorem{theorem}{Theorem}
\newtheorem{proposition}[theorem]{Proposition}%

\theoremstyle{thmstyletwo}%
\newtheorem{example}{Example}%

\theoremstyle{thmstylethree}%

\begin{document}

\title[Characterization of multi-way binary tables with uniform margins and fixed correlations]{Characterization of multi-way binary tables with uniform margins and fixed correlations}


\author[1]{\fnm{Roberto} \sur{Fontana}}\email{roberto.fontana@polito.it}

\author[2]{\fnm{Elisa} \sur{Perrone}}\email{e.perrone@tue.nl}

\author*[3]{\fnm{Fabio} \sur{Rapallo}}\email{fabio.rapallo@unige.it}

\affil[1]{\orgdiv{Department of Mathematical Sciences}, \orgname{Politecnico di Torino}, \orgaddress{\street{Corso Duca degli Abruzzi 24}, \city{Torino}, \postcode{10129}, \country{Italy}}}

\affil[2]{\orgdiv{Department of Mathematics and Computer Science}, \orgname{Eindhoven University of Technology}, \orgaddress{\street{Groene Loper 3}, \city{Eindhoven}, \postcode{5612 AE}, \country{the Netherlands}}}

\affil*[3]{\orgdiv{Department of Economics}, \orgname{Universit\`a di Genova}, \orgaddress{\street{via Vivaldi 5}, \city{Genova}, \postcode{16126}, \country{Italy}}}


\abstract{In many applications involving binary variables, only pairwise dependence measures, such as correlations, are available. However, for multi-way tables involving more than two variables, these quantities do not uniquely determine the joint distribution, but instead define a family of admissible distributions that share the same pairwise dependence while potentially differing in higher-order interactions. In this paper, we introduce a geometric framework to describe the entire feasible set of such joint distributions with uniform margins. We show that this admissible set forms a convex polytope, analyze its symmetry properties, and characterize its extreme rays. These extremal distributions provide fundamental insights into how higher-order dependence structures may vary while preserving the prescribed pairwise information. Unlike traditional methods for table generation, which return a single table, our framework makes it possible to explore and understand the full admissible space of dependence structures, enabling more flexible choices for modeling and simulation. We illustrate the usefulness of our theoretical results through examples and a real case study on rater agreement.}



\maketitle

\section{Introduction}

Modeling the dependence structure among categorical and binary variables is a central problem in statistics and is crucial in the analysis of contingency tables, with important uses in simulation, inference, and model selection (see, for example, \cite{agresti:12,rudas2018lectures}). Investigating multivariate dependence is particularly valuable in applications such as rater agreement analysis, where measures of association become especially challenging to handle (see, e.g., \cite{voneyebook}).
When dealing with two categorical or binary variables, the situation is relatively well understood: Correlation is related to the odds ratio and uniquely determines the joint distribution when the marginals are fixed \cite{agresti:12}. However, the picture changes drastically in higher dimensions. When the number of variables increases to three or more, the correlations (or marginal odds ratios) no longer uniquely define the joint distribution \cite{agresti:12}, {\cite{FONTANAJMVA}. Instead, they identify a family of admissible distributions that all share the same pairwise dependence structure but differ in their higher-order interactions.

In practice, methods based on log-linear models \cite{Hammond12082024} or the Multivariate Iterative Proportional Fitting Procedure (MIPFP) \cite{barthelemy2018mipfp} are commonly used to generate probability tables with prescribed marginal distributions and pairwise correlations. However, these approaches intrinsically select a single distribution from the set of feasible joint distributions compatible with the specified pairwise correlation. 
This selection often occurs without explicitly acknowledging the multiplicity of admissible solutions and without offering insight into how higher-order dependencies are implicitly fixed by the chosen modeling strategy. As a result, these methods may impose unintended structural constraints. This limitation becomes particularly relevant in applications such as rater agreement analysis, where MIPFP often forces higher-order odds ratios to be one, or log-linear models impose a specific hierarchical dependence structure, potentially leading to unrealistic or overly restrictive representations of joint agreement.

The goal of this paper is to characterize the entire feasible set of joint distributions with fixed pairwise dependence (expressed through correlations or marginal odds ratios) and uniform margins. 
This work complements recent studies on odds ratios and zero-pattern compatibility in contingency tables \cite{fontana2025IES}, by shifting the focus to the geometric structure induced by fixed pairwise correlations.
By focusing on the full admissible set rather than a single model-generated table, we bridge dependence modeling and log-linear representations, offering a framework that enables a more informed and conscious selection of a probability table within the admissible family.
Our approach adopts a geometric perspective: We study the feasible region of admissible distributions as a convex polytope, whose vertices and structure encode all possible joint distributions consistent with the given constraints. 
The restriction to the class of distributions with uniform margins is inspired by recent developments in copula theory for discrete random vectors \cite{Perrone_2019, geenens2020copula,Perrone2021,Perrone2022, Perrone2024,kojadinovic2024}, whose central idea is to depict the dependence disregarding the marginal behavior. The restriction to tables with uniform margins has clear benefits for our approach: By focusing on uniform margins, we obtain a clean geometric representation of the dependence structure and uncover elegant symmetries in the space of joint distributions, which are instrumental in describing and characterizing the entire feasible set.
This geometric characterization is particularly relevant for applications where pairwise dependence information is available or desired, but higher-order interactions are uncertain or flexible. 
Examples include rater agreement analysis, where the true nature of joint agreement beyond the pairwise level may vary, as well as simulation and missing data imputation tasks, where exploring multiple feasible probability tables is both meaningful and necessary. 
As a matter of fact, in the analysis of rater agreement it is well known that marginal non-homogeneity can make difficult the interpretation of agreement measures \cite{feinstein|cicchetti:90}, and thus a solution already explored in the literature \cite{agresti|etal:95} is to transform the original table into an equivalent one with homogeneous, i.e., uniform, margins. 
By characterizing the entire admissible class of binary tables with uniform margins, our framework offers a more informative and practically useful view of the dependence structure among binary variables. 
While in this work we focus on binary tables, we recognize extending the theory to multi-way non-binary tables as an important avenue for future research. 

The remainder of the paper is structured as follows. Section~\ref{sec:background} introduces the notation, problem formulation, connections with log-linear models, and existing approaches to generate tables with prescribed correlation. Section~\ref{sec:framework} presents the geometric characterization of the feasible set, including its polyhedral structure. Section~\ref{sec:examples} illustrates the method through examples and shows its applicability in practice for rater agreement. Finally, Section~\ref{sec:discussion} concludes the paper with a summary and possible directions for future work.

\section{Background and related work}\label{sec:background}

In this section, we first introduce the notation and recall key concepts used to describe dependence among binary variables. In particular, we review the notions of pairwise correlations and odds ratios, discuss their relationships, and highlight their 
roles in determining the dependence structure of a joint distribution.

We consider $d \geq 2$ binary random variables $X_1,\ldots,X_d$, with joint cell probabilities
\[
p_\alpha = \mathbb{P}(X_1=\alpha_1,\ldots,X_d=\alpha_d), \qquad \alpha=(\alpha_1,\ldots,\alpha_d)\in\{0,1\}^d.
\]
When needed, we assume the set $\{0,1\}^d$ in lexicographic order, and then we can refer to the $k$-th cell of $T$, $k=1,\ldots, 2^d$. We observe that the $k$-th cell of a table is the one corresponding to $\alpha_k=((\alpha_{k})_1, \ldots, (\alpha_{k})_d)$ such that $\sum_{j=1}^{d} (\alpha_{k})_j 2^{d-j}=k-1, \; k=1,\ldots,2^d$, i.e., $\alpha_k$ is the binary representation of $k-1$. Thus, a table $T$ is described by the $2^d$-vector $p=(p_\alpha, \alpha \in \{0,1\}^d)$. To simplify the notation, we denote $\{0,1\}^d$ by $\mathcal{D}$, $E[X_i]$ by $\mu_i$, and $E[X_i X_j]$ by $\mu_{ij}$. Given $\alpha=(\alpha_1,\ldots,\alpha_d)$ we denote by $1-\alpha$ the vector $(1-\alpha_1,\ldots,1-\alpha_d)$.

There are $d$ univariate margins $m_i, i=1,\ldots,d$ defined as
\[
m_i=(m_{i}^{0},m_{i}^{1})=\left(\sum_{\alpha \in \mathcal{D}, \alpha_i=0} p_\alpha, \sum_{\alpha \in \mathcal{D}, \alpha_i=1} p_\alpha \right).
\]
We observe that $m_{i}^{1}=\mu_i$ and $m_{i}^{0}=1-\mu_i$, $i=1,\ldots,d$.

We also consider the $\binom{d}{2}$ bivariate margins
\[
m_{ij}=
\begin{pmatrix}
m_{ij}^{00} & m_{ij}^{01}\\
m_{ij}^{10} & m_{ij}^{11}
\end{pmatrix}\, ,
\]
where 
\[
m_{ij}^{k_1 k_2}=\sum_{\alpha \in \mathcal{D}, \alpha_i=k_1, \alpha_j=k_2, } p_\alpha,
\]
$k_1,k_2 \in \{0,1\}$, and ${i,j=1,\ldots,d, i<j}$.

We observe that $m_{ij}^{11}=\mu_{ij}$, and if we fix the univariate margins $m_i$ and $m_j$ to $\overline{m}_i=(\overline{m}_i^0,\overline{m}_i^1)$ and $\overline{m}_j=(\overline{m}_j^0,\overline{m}_j^1)$, respectively, we can express $m_{ij}$ as a function of $m_{ij}^{11}$:
\begin{equation} \label{eq:2mc}
m_{ij}=
\begin{pmatrix}
\overline{m}_{i}^{0}-\overline{m}_{j}^{1}+ m_{ij}^{11} & \quad \overline{m}_{j}^{1}- m_{ij}^{11}\\
\overline{m}_{j}^{1}-m_{ij}^{11} & m_{ij}^{11}
\end{pmatrix}	
\end{equation}

For any $i\neq j \in \{1,\ldots,d\}$, the pairwise correlation between $X_i$ and $X_j$ depends only on the marginal probabilities and the marginal $2\times2$ sub-table obtained by summing over all other variables. It can be expressed as
\begin{flalign}
\label{eq:pair_corr_general}
\mathrm{Cor}(X_i, X_j)
&=  \dfrac{\mathbb{P}(X_i=1,X_j=1) - \mathbb{P}(X_i=1)\mathbb{P}(X_j=1)}
{\sqrt{\mathbb{P}(X_i=1)(1-\mathbb{P}(X_i=1))\,\mathbb{P}(X_j=1)(1-\mathbb{P}(X_j=1))}} \nonumber\\[4pt]
&= 
\dfrac{{m}_{ij}^{11} - m_{i}^{1}m_{j}^{1}}
{\sqrt{m_{i}^{1} m_{i}^{0} m_{j}^{1} m_{j}^{0}  }}  .
\end{flalign}

\medskip
For any pair $(X_i,X_j)$, we define conditional odds ratios by fixing the values of the remaining $(d-2)$ variables. Let $\alpha'\in\{0,1\}^{d-2}$ denote the configuration of the other variables. Then, the conditional odds ratio of $(X_i,X_j)$ given $X_{-(i,j)}=\alpha'$ is
\begin{equation} \label{oreq_general}
\omega_{\alpha'}^{ij}
=
\frac{p_{\alpha_1}\,p_{\alpha_2}}{p_{\alpha_3}\,p_{\alpha_4}},
\end{equation}
where the four index vectors $\alpha_1,\ldots,\alpha_4\in\{0,1\}^d$ share the same entries as the $\alpha'$ outside positions $i$ and $j$, while their $(i,j)$ entries are, respectively,
\[
\alpha_1(i,j)=(1,1),\quad
\alpha_2(i,j)=(0,0),\quad
\alpha_3(i,j)=(1,0),\quad
\alpha_4(i,j)=(0,1).
\]
For example, in $d=4$, fixing $(X_2,X_4)=(0,1)$ yields:
\[
\omega_{01}^{13}=
\frac{p_{0001}\,p_{1011}}{p_{0011}\,p_{1001}}.
\]

In a $d$-way binary table, there are $\binom{d}{2}2^{d-2}$ such conditional odds ratios, though these are not all independent.

The marginal odds ratio for a pair $(X_i,X_j)$ is obtained by summing over all other variables:
\begin{equation} \label{eq:oddmarg}
 \omega_M^{ij}
=
\frac
{\sum_{\alpha:\alpha_i=\alpha_j=1}p_{\alpha} \;\sum_{\alpha:\alpha_i=\alpha_j=0}p_{\alpha}}
{\sum_{\alpha:\alpha_i=1,\alpha_j=0}p_{\alpha} \;\sum_{\alpha:\alpha_i=0,\alpha_j=1}p_{\alpha}} =
\frac
{m_{ij}^{11} m_{ij}^{00}} 
{m_{ij}^{10} m_{ij}^{01}} \, .   
\end{equation}
    

\medskip

In this setting, the dependence structure among the variables can be summarized in different ways, most commonly through pairwise correlations or through marginal and conditional odds ratios. 
These notions coincide in the bivariate case, where, given the margins, either of them uniquely determines the joint distribution, but diverge in higher dimensions, where they only constrain a family of admissible distributions. 
In the following sections, we examine separately the bivariate and multivariate cases to highlight these differences, and we discuss how correlations and odds ratios determine (or fail to determine) the underlying joint distribution.
Along the way, we highlight how different odds ratio structures relate to the full joint distribution through their copula representation and connect naturally to the main problem considered in this study \cite{fontana2026ijar, geenens2020copula, kojadinovic2024}.
This perspective will serve as a conceptual foundation for our later geometric analysis of the feasible set of distributions with fixed pairwise dependence and uniform margins.

\subsection{Two-way binary tables}

We begin with the simple case of binary $2\times2$ tables. In this case, the dependence between $X_1$ and $X_2$ can be described either through the correlation or by the odds ratio, and these two measures are closely related.

The Pearson correlation between $X_1$ and $X_2$ can be written as
\begin{flalign}
\label{eq:pair_corr_2}
\mathrm{Cor}(X_1, X_2)
&=  \dfrac{\mathbb{P}(X_1=1,X_2=1) - \mathbb{P}(X_1=1)\mathbb{P}(X_2=1)}
{\sqrt{\mathbb{P}(X_1=1)(1-\mathbb{P}(X_1=1))\mathbb{P}(X_2=1)(1-\mathbb{P}(X_2=1))}} \nonumber \\
&=
\dfrac{p_{11} - m_{1}^1 m_{2}^1}
{\sqrt{m_{1}^1(1-m_{1}^1) \, m_{2}^1(1-m_{2}^1)}} \nonumber\\[3pt]
&=
\dfrac{p_{11}p_{00} - p_{10}p_{01}}
{\sqrt{m_{1}^1m_{1}^0  m_{2}^1 m_{2}^0}}.
\end{flalign} 

The concept of odds ratio is also a viable way to characterize the dependence for $2 \times 2$ tables. Specifically, in a $2\times2$ table there is a single odds ratio, defined as
\[
\omega^{12} = \frac{p_{11}p_{00}}{p_{10}p_{01}}.
\]
The numerator of the correlation, i.e. the covariance $\mathrm{Cov}(X_1,X_2)$, coincides with $p_{11}p_{00} - p_{10}p_{01}$ and is directly related to the odds ratio, in the sense that $\omega^{12} = 1$ if and only if $X_1$ and $X_2$ are independent, which in this case corresponds to $\mathrm{Cor}(X_1,X_2)=0$. More generally, $\omega^{12}>1$ implies a positive association and $\omega^{12}<1$ a negative one \cite{agresti:12, rudas2018lectures}.

\medskip

\noindent \textbf{Dependence through odds ratios under uniform margins.}
When the margins of the table are transformed to be uniform, the resulting table is uniquely determined by the odds ratio (or equivalently by the correlation) in line with recent results in the copula theory \cite{geenens2020copula}. Indeed, in the bivariate binary case there exists a unique probability distribution ${p}$ with uniform margins, i.e.,
\[
\overline{m}_{1}^1 = \overline{m}_{2}^1 = \frac{1}{2},
\]
such that ${p}$ preserves the original odds ratio $\omega^{12}$. This means that the odds ratio (and equivalently the correlation) fully characterizes the dependence structure once the margins are fixed to be uniform. In other words, in dimension $d=2$, the feasible set of distributions with fixed pairwise dependence and uniform margins reduces to a single point.
Such a single point is also constrained to be symmetric along the diagonals, i.e., with ${p}_{00}={p}_{11}$ e ${p}_{01}={p}_{11} = \frac{1}{2} - {p}_{00}= \frac{1}{2} -{p}_{11} $. Additionally, in this a case, by a simple substitution in Eq~\eqref{eq:pair_corr_2}, we can re-write $\mathrm{Cor}(X_1,X_2)$ as a linear function of $p_{11}$, namely, $\mathrm{Cor}(X_1,X_2)=4 p_{11} - 1$.

For $d \geq 3$, the solution set is nontrivial but still preserves the symmetry property observed in the bivariate case. This highlights the main conceptual distinction between the bivariate and multivariate frameworks and, in higher dimensions, justifies our emphasis on describing the full admissible set rather than focusing on a single representative distribution.

\subsection{Multi-way binary tables}
When $d \ge 3$, the relationship between correlations and odds ratios becomes more intricate, and fixing all pairwise correlations (or equivalently all marginal odds ratios) is no longer sufficient to uniquely determine the joint distribution. 

To illustrate the situation, let us consider $d=3$, i.e., a random vector $(X_1,X_2,X_3)$. Using Eq.~\eqref{eq:pair_corr_general} the pairwise correlation between $X_1$ and $X_2$ is given by
\begin{flalign*}
\label{eq:pair_corr_3way}
\mathrm{Cor}(X_1, X_2)
&= \dfrac{m_{12}^{11}m_{12}^{00} - m_{12}^{10} m_{12}^{01}}
{\sqrt{m_1^1 m_1^0 m_2^1 m_2^0}}.
\end{flalign*}

Analogous formulas hold for $\mathrm{Cor}(X_1,X_3)$ and $\mathrm{Cor}(X_2,X_3)$.

\medskip

\noindent \textbf{Conditional and marginal odds ratios.}
In the case $d=3$, there are two fundamentally different types of odds ratios.

\begin{itemize}
\item[\emph{(i)}] \emph{Conditional odds ratios.} For each pair $(X_i,X_j)$, we may condition on a value of the remaining variable. According to Eq.~\eqref{oreq_general} this yields six conditional odds ratios:
\begin{equation} \label{orsys}
\begin{split}
\omega_{0}^{23}= \frac {p_{000}p_{011}} {{p_{001}p_{010}}}, \qquad 
\omega_{1}^{23}= \frac {p_{100}p_{111}} {{p_{101}p_{110}}}, \\[4pt]
\omega_{0}^{13}= \frac {p_{000}p_{101}} {{p_{001}p_{100}}}, \qquad 
\omega_{1}^{13}= \frac {p_{010}p_{111}} {{p_{011}p_{110}}}, \\[4pt]
\omega_{0}^{12}= \frac {p_{000}p_{110}} {{p_{010}p_{100}}}, \qquad 
\omega_{1}^{12}= \frac {p_{001}p_{111}} {{p_{011}p_{101}}}.
\end{split}
\end{equation}
Only four of these six constraints are independent.

\item[\emph{(ii)}] \emph{Marginal odds ratios.} We may instead consider the odds ratios of the $2\times2$ tables obtained by summing over (collapsing) the remaining variable according to Eq.~\eqref{eq:oddmarg}:
\begin{equation} \label{ormarg}
\omega_M^{12}= 
\frac
{m_{12}^{11} m_{12}^{00}} 
{m_{12}^{10} m_{12}^{01}}, \qquad
\omega_M^{23}=\frac
{m_{23}^{11} m_{23}^{00}} 
{m_{23}^{10} m_{23}^{01}}, \qquad 
\omega_M^{13}=\frac
{m_{13}^{11} m_{13}^{00}} 
{m_{13}^{10} m_{13}^{01}}.
\end{equation}
\end{itemize}

The conditional and marginal odds ratios do not coincide in general. 
All conditional odds ratios between a pair $(X_i,X_j)$ are possible to be equal to 1 (which suggests conditional independence), while the marginal odds ratio $\omega_M^{ij}$ indicates strong association. This classical phenomenon is known as Simpson’s paradox and emphasizes that marginal odds ratios describe associations in the collapsed table, while conditional odds ratios reflect dependence conditioned on other variables \cite{agresti:12,rudas2018lectures}.

\medskip

\noindent \textbf{Correlations, marginal odds ratios, and dependence.}
The correlation between two binary random variables $X_i$  and $X_j$ is directly related to the corresponding {marginal} odds ratio $\omega_M^{ij}$, but not to the conditional odds ratios. Hence, specifying all pairwise correlations (or equivalently all marginal odds ratios) does not determine the conditional odds ratios or any higher-order interaction terms even when the margins are uniform. 
%
This point is crucial: while conditional odds ratios and uniform margins uniquely determine the full joint distribution in arbitrary dimension (see, for example, \cite{geenens2020copula, kojadinovic2024,fontana2026ijar} and references therein), marginal odds ratios (i.e., pairwise correlations) do not. Hence, fixing all pairwise correlations and the margins determines not a unique joint distribution, but a whole family of distributions. The structural analysis and geometric characterization of this feasible family are the focus of Section~\ref{sec:framework}. In the reminder of this section we provide a brief overview of the main algorithms for sampling from correlated binary distributions, to emphasize the relevance of our geometric description.

\subsection{Algorithms for sampling from correlated binary random variables}

Generating multivariate binary random variables with a given correlation structure is a crucial task in simulation studies, and a rich variety of methods have been described in the literature to tackle this problem. Random generation plays a major role in, e.g., inference and model validation. Here we review the main different approaches that show close connections with our theory. 

The paper \cite{lee:93} proposes a linear-programming method for constructing the joint distribution by imposing linear constraints on the cell probabilities or minimizing certain linear functions. This approach can be readily implemented with standard Linear Programming solvers. However, as the table’s dimension grows, the large number of required linear constraints can introduce considerable arbitrariness into the resulting distribution.

In \cite{barthelemy2018mipfp}, the authors present a general procedure for transforming one multi-way table into another that has prescribed margins while preserving the original odds ratio structure. Because its two-dimensional counterpart, known as Iterative Proportional Fitting, is a classical method in statistics and numerical analysis \cite{fienberg:70}, this multidimensional extension is referred to as MIPFP. As a byproduct of the MIPFP, the authors of \cite{barthelemy2018mipfp} present a procedure to construct multivariate Bernoulli distributions with prescribed pairwise marginal odds ratios
\[
\omega_M^{ij}= \frac {\left(\sum_{\alpha_i=1,\alpha_j=1}p_\alpha\right)\left(\sum_{\alpha_i=0,\alpha_j=0}p_\alpha\right) } {\left(\sum_{\alpha_i=0,\alpha_j=1}p_\alpha\right)\left(\sum_{\alpha_i=1,\alpha_j=0}p_\alpha\right) }=
\frac
{m_{ij}^{11} m_{ij}^{00}} 
{m_{ij}^{10} m_{ij}^{01}} \, ,
\]
and simultaneously set all higher-order odds ratios in the contingency table equal to zero. This procedure is implemented in \texttt{R} using the \texttt{ObtainMultBinaryDist} function from the \texttt{mipfp} package. The approach is computationally efficient, and the desired output can be produced through a straightforward application of MIPFP. The algorithm is applicable to tables of any size, however, it is particularly tailored to handle binary tables.

Finally, in \cite{Hammond12082024}, the authors make use of the log-linear form of the joint distribution for table generation. Correlations are expressed using odds ratios, and by applying the corner parametrization along with suitable algebraic refinement steps, a valid joint distribution is obtained.For example, for a three-way binary table and under the assumption that there are no three-way interaction terms, the model is given by
\[
\log p_{ijk} = \lambda + \lambda_i^1 +  \lambda_j^2 + \lambda_k^3 + \lambda_{ij}^{12} + \lambda_{ik}^{12} + \lambda_{jk}^{13}.
\]
When this model is reparameterized in terms of odds ratios, it can be written as
\[
\log p_{ijk} = \lambda + \lambda_i^1 +  \lambda_j^1 + \lambda_k^3 + \omega_M^{12} + \omega_M^{13} + \omega_M^{23}.
\] 
The key benefit of this approach is that it can be generalized to multilevel tables with an arbitrary number of levels. Nevertheless, as the dimensionality increases, the links between correlations and odds ratios can require intricate equations, which complicates selecting parameters based on the marginal odds ratios. In the introductory section of the same paper \cite{Hammond12082024}, there is a brief review of other methods based on the theory of log-linear models.  

All of the approaches discussed above rely on choosing a single joint distribution, usually by imposing extra assumptions beyond the given pairwise dependencies. As emphasized in the introduction, our aim is instead to describe the complete collection of feasible distributions, thus offering a systematic perspective on all attainable dependence structures, which will be the subject of the next section.

\section{Geometric characterization of multi-way binary tables with uniform margins and fixed correlations}\label{sec:framework}

We now formalize the problem at the core of our theoretical developments.

\noindent\textit{\textbf{Problem statement:}} Given $d$ binary random variables with a prescribed pairwise dependence structure (expressed through correlations or marginal odds ratios), determine the full set of joint distributions with uniform margins that satisfy these constraints. In particular, describe the geometric structure of this admissible set, understand how higher-order interactions may vary within it, and identify its extremal elements.

As discussed in Section~\ref{sec:background}, restricting attention to tables with uniform margins is both mathematically and statistically advantageous. From a mathematical perspective, uniform margins induce symmetries in the table, which facilitate a clearer geometric characterization of the feasible set. From a statistical viewpoint, working with uniform margins eliminates the effect of unbalanced marginal distributions, allowing for a more direct control on the dependence structure.
To incorporate uniform margins into our framework, we now formulate the condition of uniform margins in algebraic form.
Specifically, to obtain uniform margins, the probability mass function (pmf) $p$ of a table $T$ must satisfy the following constraints
\[
m_{i}^0=m_{i}^1=\frac{1}{2} \Leftrightarrow m_{i}^0-m_{i}^1=0 \Leftrightarrow \sum_{\alpha \in \mathcal{D}, \alpha_i=0} p_\alpha- \sum_{\alpha \in \mathcal{D}, \alpha_i=1} p_\alpha=0,  \; i=1,\ldots,d \, ,
\]
that can be rewritten as
\begin{equation} \label{eq:system}
 (1_{\{\alpha \in \mathcal{D}, \alpha_i=0\}}-1_{\{\alpha \in \mathcal{D}, \alpha_i=1\}})^Tp=0, \;\;i=1,\ldots,d \, , 
\end{equation}
where $1_{A}$ is the indicator vector of $A$, $b^T$ denotes the transpose of $b$, and $p$ is the vector $p=(p_\alpha, \alpha \in \mathcal{D})$. 

If, in addition to uniform margins, we want to specify the bivariate margins, we observe that Eq.~\eqref{eq:2mc} becomes
\begin{equation} \label{eq:2mcuni}
m_{ij}=
\begin{pmatrix}
m_{ij}^{11} & \frac{1}{2}- m_{ij}^{11}\\
\frac{1}{2}-m_{ij}^{11} & m_{ij}^{11}
\end{pmatrix} \, .
\end{equation}
Then, to specify uniform margins it will be enough to specify $m_{ij}^{11}=\overline{m}_{ij}^{11}$ or, equivalently, $\mu_{ij}=\overline{\mu}_{ij}$, ${i,j=1,\ldots,d, i<j}$.
We obtain
\begin{equation*}
\sum_{\alpha \in \mathcal{D}, \alpha_i=\alpha_j=1 } p_\alpha =\overline{\mu}_{ij},
\end{equation*}
that is
\begin{equation} \label{eq:system2tmp}
\sum_{\alpha \in \mathcal{D}, \alpha_i\alpha_j=1} (\overline{\mu}_{ij}-1) p_\alpha +\overline{\mu}_{ij}\sum_{\alpha \in \mathcal{D}, \alpha_i \alpha_j=0} p_\alpha =0.
\end{equation}
We can write  Eq.~\eqref{eq:system2tmp} as
\begin{equation} \label{eq:system2}
 ((\overline{\mu}_{ij}-1)1_{\{\alpha \in \mathcal{D}, \alpha_i \alpha_j=1\}}+\overline{\mu}_{ij}1_{\{\alpha \in \mathcal{D}, \alpha_i \alpha_j=0\}})^Tp=0, \;\;i,j=1,\ldots,d, \;\; i<j. 
\end{equation}

Eq.~\eqref{eq:system} and Eq.~\eqref{eq:system2} can be put in compact form by denoting as $H_d$ its matrix of coefficients. Namely, Eq.~\eqref{eq:system} and Eq.~\eqref{eq:system2} become:

\begin{equation} \label{eq:HD}
    H_d \,p=0 \, .
\end{equation}

The polyhedral cone defined by $H_d$ is therefore
\begin{equation} \label{eq:polycone}
\mathcal{C}(H_d) = \left\{ p=(p_\alpha), p_\alpha \in \mathbb{R}, p_\alpha\geq 0, \; \alpha \in \mathcal{D},  H_d p=0,  \right\} \, . 
\end{equation}

For example, for $d=3$ we obtain:
\[
H_3 p=
\begin{pmatrix}
1 & 1 & 1 & 1 & -1 & -1 & -1 & -1  \\
1 & 1 & -1 & -1 & 1 & 1 & -1 & -1 \\
1 & -1 & 1 & -1 & 1 & -1 & 1 & -1 \\
\overline{\mu}_{12} & \overline{\mu}_{12} & \overline{\mu}_{12} & \overline{\mu}_{12} & \overline{\mu}_{12} & \overline{\mu}_{12} & (\overline{\mu}_{12}-1) & (\overline{\mu}_{12}-1) \\
\overline{\mu}_{13} & \overline{\mu}_{13} & \overline{\mu}_{13} & \overline{\mu}_{13} & \overline{\mu}_{13} & (\overline{\mu}_{13}-1) & \overline{\mu}_{13} & (\overline{\mu}_{13}-1)  \\
\overline{\mu}_{23} & \overline{\mu}_{23} & \overline{\mu}_{23} & (\overline{\mu}_{23}-1) & \overline{\mu}_{23} &  \overline{\mu}_{23} & \overline{\mu}_{23} & (\overline{\mu}_{23}-1)  \\
\end{pmatrix}p=0 \, .
\]

The solutions of the homogeneous system $H_d p = 0$ with $p_\alpha \ge 0$ form the polyhedral cone $\mathcal{C}(H_d)$ defined in Eq.~\eqref{eq:polycone}. 
An \emph{extreme ray} of this cone is a nonzero vector $\tilde{y} \in \mathcal{C}(H_d)$ that cannot be expressed as a nontrivial linear combination of other elements of the cone. 
Thus, extreme rays represent the fundamental generating directions of the cone: Every element $p \in \mathcal{C}(H_d)$ can be written as a nonnegative linear combination of its extreme rays. Identifying these generating rays is therefore essential for understanding the structure of the feasible set.
The extreme rays $\tilde{y}_i=(\tilde{y}_{i\alpha}, i=1,\ldots,n_d, \alpha \in \mathcal{D}, \tilde{y}_\alpha \geq 0)$ of the system $H_d y = 0$, where $y=(y_\alpha, \alpha \in \mathcal{D}, y_\alpha \geq 0)$, coincide with the extreme rays of the polyhedral cone $\mathcal{C}(H_d)$ introduced in Eq.~\eqref{eq:polycone} (see \cite{dahl97,barvinok:02}). They can be computed using software for commutative algebra, such as \texttt{4ti2} \cite{4ti2}, which is used in this work. Then the corresponding extreme pmfs $r_i=(r_{i\alpha}, i=1,\ldots,n_d, \alpha \in \mathcal{D})$ are computed by simple normalization $r_{i\alpha}=\frac{\tilde{y}_{i\alpha}}{\sum_{\alpha \in \mathcal{D}} \tilde{y}_{i\alpha}}, i=1,\ldots,n_d, \alpha \in \mathcal{D}$. The intersection of the polyhedral cone $\mathcal{C}(H_d)$ with the simplex
\[
\Delta=\left\{ p_\alpha \in \mathbb{R}^{(2^d)} \ : \ p_\alpha \ge 0, \ \sum_\alpha p_\alpha=1 \right\}
\]
is the convex polytope $\mathcal{F}({H_d})$ whose extreme points are the extreme pmfs $r_i, i=1,\ldots,n_d$: 
\begin{equation} \label{eq:feasibleset}
\mathcal{F}({H_d})=\left\{ p=(p_\alpha), p_\alpha \in \mathbb{R}, p_\alpha\geq 0, \; \alpha \in \mathcal{D},  H_d p=0,  \sum_{\alpha \in \mathcal{D}} p_\alpha =1\right\} \, .
\end{equation}

It is worth noting that, given the extreme pmfs $r_i, i=1,\ldots,n_d$ where both $r_i, i=1,\ldots,n_d$ and $n_d$ depend on the dimension $d$ and the specified univariate and bivariate margins of the multi-way table, any pmf $p$ that satisfies Eq.~\eqref{eq:HD} can be written as 
\[
p=\sum_{i=1}^{n_d} \theta_i r_i
\]
with $\theta_i\geq0, i=1,\ldots,n_d$ and $\sum_{i=1}^{n_d} \theta_i=1$.
With the three propositions below, we study the geometric structure of the feasible set of probability distributions in Eq.~\eqref{eq:feasibleset}. 

\begin{proposition} \label{prop:sym}
Let $(X_1,\ldots,X_d)$ be a multivariate Bernoulli random variable with distribution $p=(p_\alpha, \alpha \in \mathcal{D})$, uniform margins and given second-order moments $\mu_{12},\mu_{13},\ldots,\mu_{d-1,d}$. Then, the multivariate Bernoulli random variable $(X'_1,\ldots,X'_d)$ with distribution $p'=(p'_{\alpha}=p_{1-\alpha}, \alpha \in \mathcal{D})$ also has uniform margins and the same second-order moments. 
\end{proposition}
\begin{proof}
Let $(X'_1,\ldots,X'_d)$ a multivariate Bernoulli random variable with distribution $p'$. We get
\[
\mu'_i=E[X'_i]=\sum_{\alpha \in \mathcal{D}, \alpha_i=1} p'_\alpha=\sum_{\alpha \in \mathcal{D}, \alpha_i=1} p_{1-\alpha}=\sum_{\alpha \in \mathcal{D}, \alpha_i=0}p_\alpha=1-E[X_i]=\frac{1}{2}, \; 1,\ldots,d.
\]
It follows that $(X'_1,\ldots,X'_d)$ has uniform margins. Similarly for the second-order moments we obtain
\begin{eqnarray*}
  \mu'_{ij}=E[X'_i X'_j]=\sum_{\alpha \in \mathcal{D}, \alpha_i=\alpha_j=1} p'_\alpha=\sum_{\alpha \in \mathcal{D}, \alpha_i=\alpha_j=1}p_{1-\alpha}=\sum_{\alpha \in \mathcal{D}, \alpha_i=\alpha_j=0} p_\alpha= \\
=E[(1-X_i)(1-X_j)]=1-E[X_i]-E[X_j]+E[X_iX_j]=E[X_iX_j]=\mu_{ij}, \\
i,j=1,\ldots,d, i<j.  
\end{eqnarray*}
\end{proof}

In the following result, we prove that the symmetry of the feasible solutions can be extended to the extreme pmfs, and therefore the extreme pmfs come in pairs with a precise combinatorial structure. The proof is based on a geometric definition of extreme ray, see \cite[p.~65]{barvinok:02}

\begin{proposition} \label{prop:symrays}
If $r=(r_\alpha)$ is an extreme pmf of the polyhedral cone for the problem in Eq.~\eqref{eq:polycone}, then also $r'=(r_{1-\alpha})$ is an extreme pmf.
\end{proposition}
\begin{proof}
For the sake of simplicity, we work on the unnormalized polyhedral cone defined by Eq.~\eqref{eq:polycone}. We know by Prop.~\ref{prop:sym} that $r'$ is a feasible distribution. We need to prove that $r'$ is an extreme pmf. We know that $r$ is an extreme pmf in $\mathcal{F}(H_d)$ if and only if writing 
\[
r = p_1 + p_2 \qquad \qquad p_1,p_2 \in \mathcal{C}(H_d)
\]
then $p_1=\theta_1 r$ and $p_2=\theta_2 r$, $\theta_1,\theta_2\ge 0$. 

Now, consider $r'$ and proceed by contradiction. Suppose that $r'$ is not an extreme pmf. Thus, we can write
\begin{equation}\label{eq:contr}
r' = p_1 + p_2 \qquad \qquad p_1,p_2 \in \mathcal{C}(H_d)
\end{equation}
with $p_1$ and $p_2$ not simultaneously multiple of $r'$. Symmetrizing again the previous equation and noticing that $(r')' = r$, we have
\[
r = (r')' = p_1' + p_2'
\]
Since $p_1'$ and $p_2'$ belong to $\mathcal{C}(H_d)$ by Prop.~\ref{prop:sym}, and $r$ is an extreme pmf, we get that both $p_1'$ and $p_2'$ must be nonnegative multiples of $r$:
\[
p_1' = \theta_1 r, \qquad p_2' = \theta_2 r
\]
for some nonnegative $\theta_1,\theta_2$. Applying again the symmetrization with respect to the indices we obtain:
\[
p_1 = \theta_1 r', \qquad p_2 = \theta_2 r'
\]
and thus both $p_1$ and $p_2$ are nonnegative multiples of $r'$, in contradiction with Eq.~\eqref{eq:contr}.
\end{proof}

The last result concerns the symmetries in the log-linear representation of the pmfs in the feasible set. That symmetry lies on the special features of the zero-mean parametrization of the model. We recall that, writing a non-negative probability distribution in the saturated log-linear representation, we have
\[
\log p_{i_1 \dots i_d}
=
\lambda
+
\sum_{\emptyset \neq S \subseteq \{1,\dots,d\}}
\lambda^S_{i_S} \, ,
\]
where the sum runs over all nonempty subsets $S$ of $\{1,\dots,d\}$ and each $i_S$ is the subset of the indices identified by $S$. The zero-mean parametrization assumes that the sum of the parameters $\lambda^S_{i_S}$ is zero over all one-dimensional indices in $S$. For instance, in the two-way case we write:
\[
\log p_{\alpha_1 \alpha_2} = \lambda + \lambda_{\alpha_1}^{1}+ \lambda_{\alpha_2}^{2} + \lambda_{\alpha_1 \alpha_2}^{12} 
\]
with $\lambda^i_0+\lambda^i_1=0$, $i=1,2$ and $\lambda^{12}_{00}+\lambda^{12}_{01}=\lambda^{12}_{10}+\lambda^{12}_{11}=\lambda^{12}_{00}+\lambda^{12}_{00}=\lambda^{12}_{01}+\lambda^{12}_{11}=0$.

\begin{proposition} \label{prop:symloglin}
Using the zero-mean parametrization of the saturated log-linear model, two nonnegative symmetric pmfs $p$ and $p'$ have the same log-linear parameters in absolute value, with the same sign for even orders of interaction and with the opposite sign for odd orders of interaction.
\end{proposition}
\begin{proof}
It is enough to consider the saturated log-linear form of $p$:
\[
\log p_{i_1 \dots i_d}
=
\lambda
+
\sum_{\emptyset \neq S \subseteq \{1,\dots,d\}}
\lambda^S_{i_S} \, ,
\]
Noticing that mapping $p$ into $p'$ changes all the indices in $\alpha$ it is immediate to conclude that 
\[
\log p_{i_1 \dots i_d}
=
\lambda
+
\sum_{\emptyset \neq S \subseteq \{1,\dots,d\}}
(-1)^{|S|}\lambda^S_{i_S} \, ,
\]
which proves the proposition.
\end{proof}

We remark that the above proposition needs a strictly positive distribution. However, we can apply it to the extreme pmfs by adding a small positive constant to avoid zero probabilities. We will see this issue in the example below.
 
\begin{example}  \label{ex:running}
Consider the $3$-way contingency table $T_0$ whose pmf $p_0$ is reported in Table \ref{tab:marg01} together with its marginal odds ratios $\omega_M^{12}$, $\omega_M^{13}$, and $\omega_M^{23}$. 

\begin{table}[b]
	\centering
    \caption{Table $T_0$ for Example \ref{ex:running} reporting $p_0$ (exact value) and its marginal odds ratios (rounded to 2 decimal digits). Data from \cite{fontana2026ijar}.}
	\label{tab:marg01}
		\begin{tabular}{lrrrrrrrrrrrr}   \toprule
pmf & $000$ & $001$ & $010$ & $011$ & $100$ & $101$ & $110$ & $111$ & $ \, $ & $\omega_M^{12}$ &
$\omega_M^{13}$ & $\omega_M^{23}$ \\
\midrule
$p_0$ & $0.1$ & $0.05$ & $0.3$ & $0.2$ & $0.1$ & $0.05$ & $0.15$ & $0.05$ & $ \, $ & $0.40$ & $0.64$ & $1.11$ \\  
\bottomrule
\end{tabular}
\end{table}

Our goal is to determine the convex polytope of the pmfs with uniform margins and the same marginal odds ratios. As detailed in the previous sections, for two-way tables there is a one-to-one relation between the marginal odds ratio $\omega_M^{12}$ and the value of the pmf with uniform margins corresponding to the $(1,1)$ cell, that is $p_{11}$:
\begin{equation} \label{eq:ormp11}
p_{11}=\frac{\sqrt{\omega_M^{12}}}{2\left(\sqrt{\omega_M^{12}}+1\right)}.    
\end{equation}
For three-way tables, given $\omega_M^{12}$, $\omega_M^{13}$, and $\omega_M^{23}$, the equation Eq.~\eqref{eq:ormp11} can be used to determine the second-order moments $\mu_{12}=p_{110}+p_{111}$, $\mu_{13}=p_{101}+p_{111}$, and $\mu_{23}=p_{011}+p_{111}$, respectively. In our example, we obtain:
\[
\mu_{12}=0.194, \qquad \mu_{13}=0.222, \qquad \mu_{23}=0.257.
\]
Building the corresponding $H_3$ matrix and using \texttt{4ti2}, \cite{4ti2}, we obtain the extreme rays that, after normalization, provide the extreme pmfs of the polytope that contains all the pmfs with uniform margins and the same marginal odds ratio of the original table $p_0$. The extreme pmfs are reported in Table \ref{tab:erex}. In the three-way binary case, the feasible set is a segment in the simplex and the extreme pmfs are just its endpoints. 

\begin{table}[h]
	\centering \caption{Extreme pmfs $r_1$ and $r_2$ of the convex polytope of the pmfs with uniform margins and marginal odds ratios of $p_0$ for Example \ref{ex:running}, all values rounded to 3 decimal places.}
	\label{tab:erex}
		\begin{tabular}{lrrrrrrrr} \toprule
pmf & $000$ & $001$ & $010$ & $011$ & $100$ & $101$ & $110$ & $111$ \\
\midrule
$r_1$ & $0$ & $0.194$ & $0.222$ & $0.084$ & $0.257$ & $0.05$ & $0.021$ & $0.173$ \\ 
$r_2$ & $0.173$ & $0.021$ & $0.05$ & $0.257$ & $0.084$ & $0.222$ & $0.194$ & $0$ \\ 
\bottomrule
\end{tabular}
\end{table}

Let us define a generic probability in the feasible set as the mixture $p = \alpha r_1 + (1-\alpha) r_2$, with $\alpha \in [0,1]$. It is straightforward to verify that the three-dimensional odds ratio
\[
{\omega}^{123} = \frac {p_{000}p_{011}p_{101}p_{110}} {p_{111}p_{100}p_{010}p_{001}} 
\]
tends to $0$ as $\alpha$ approaches $1$ and diverges to $+\infty$ as $\alpha$ approaches $0$, thereby spanning all possible three-way dependence structures.

Additionally, the two extreme pmfs of Table~\ref{tab:erex} are symmetric in the sense of Prop.~\ref{prop:symrays}.
To clarify the symmetry in the log-linear representation stated in Prop.~\ref{prop:symloglin}, we have considered the log-linear representation 
\[
\log r_{\alpha_1 \alpha_2 \alpha_3} = \lambda + \lambda_{\alpha_1}^{1}+ \lambda_{\alpha_2}^{2} + \lambda_{\alpha_3}^{3} +  \lambda_{\alpha_1 \alpha_2}^{12} +  \lambda_{\alpha_1 \alpha_3}^{13} +  \lambda_{\alpha_2 \alpha_3}^{23} +
 \lambda_{\alpha_1 \alpha_2 \alpha_3}^{123} \, ,
\]
where the values of $r_{i,\alpha}$ have been set to $r_{i,\alpha}+10^{-8}$ to avoid numerical problems in the computation of logarithms. In addition to the zero-mean parametrization, we also considered the corner parametrization, to show that the symmetry is a peculiarity of zero-mean parametrization. The corner parametrization is a well established method to obtain a full rank model matrix. It assumes that all parameters are zero except for $i_S=1$ (i.e., all the indices in $S$ are equal to $1$). In this three-way example this means that only $\lambda$, $\lambda^{1}_1$, $\lambda^{2}_1$, $\lambda^{3}_{1}$, $\lambda_{11}^{12}$, $\lambda_{11}^{13}$, $\lambda_{11}^{23}$, $\lambda_{11}^{13}$, $\lambda_{11}^{23}$ and $\lambda_{111}^{123}$ are non-zero. 

The values of the log-linear parameters associated to $r_1$ and $r_2$ for both parametrizations are reported in Table \ref{tab:betaes1}.

%

\begin{table}[h]
	\centering
    	\caption{Parameters of the log-linear representations of the extreme pmfs $r_1$ and $r_2$ from Table \ref{tab:erex}, all values rounded to 2 decimal places.}
	\label{tab:betaes1}
		\begin{tabular}{llrrrrrrrr} \toprule
parametrization & pmf & $\lambda$ & $\lambda_1^1$ & $\lambda_1^2$ & $\lambda_1^3$ & $\lambda_{11}^{11}$ & $\lambda_{11}^{13}$ & $\lambda_{11}^{23}$ & $\lambda_{111}^{111}$ \\
\midrule
corner & $r_1$ & $-18.42$  & $17.06$  & $16.92$ & $16.78$ & $-19.41$ & $-18.42$ & $-17.75$  & $21.49$ \\ 
corner & $r_2$ & $-1.76$ & $-0.72$ & $-1.24$ & $-2.10$ & $2.08$ & $3.07$ & $3.74$ & $-21.49$ \\  \midrule
zero-mean & $r_1$ & $-4.25$ & $1.76$ & $1.85$ & $2.03$ & $-2.17$ & $-1.92$ & $-1.75$  & $2.69$ \\ 
zero-mean & $r_2$ & $-4.25$ & $-1.76$ & $-1.85$ & $-2.03$ & $-2.17$ & $-1.92$ & $-1.75$  & $-2.69$   \\ \bottomrule
\end{tabular}
\end{table}

If we fix the observed margins instead of the uniform margins, the symmetry disappears. Now our goal is to determine the convex polytope of the pmfs with the same margins and marginal odds ratios of $p_0$. Building the corresponding $H_3$ matrix and using again \texttt{4ti2}, \cite{4ti2}, we obtain the extreme rays that, after normalization, provide the extreme pmfs of the relevant feasible set. The extreme pmfs are reported in Table \ref{tab:erexobs}.  

\begin{table}[h]
	\centering \caption{Extreme pmfs $r_1$ and $r_2$ of the convex polytope of the pmfs with the same margins and marginal odds ratios of $p_0$ for Example \ref{ex:running}, all values rounded to 3 decimal places.}
	\label{tab:erexobs}
		\begin{tabular}{crrrrrrrr} \toprule
pmf & $000$ & $001$ & $010$ & $011$ & $100$ & $101$ & $110$ & $111$ \\
\midrule
$r_1$ & 0.050 &	0.100 &	0.350 &	0.150 &	0.150 &	0 &	0.100 &	0.100  \\ 
$r_2$ & 0.150 &	0 &	0.250 &	0.250 &	0.050 &	0.100 &	0.200 &	0 \\
\bottomrule
\end{tabular}
\end{table}

\end{example}
Even when exploiting the symmetry properties established in the propositions of this section, selecting a specific distribution within the feasible set remains challenging. When the dimension $d$ is moderate, i.e., up to $6$ or $7$, the proposed approach is computationally feasible: the extremal pmfs can be listed, and uniform sampling over the corresponding polytope can be performed using triangulation-based methods, such as those available in the R package \texttt{uniformly} \cite{uniformly}. As the dimension $d$ increases and the number of extreme pmfs becomes too large, other techniques, such as those implemented in the R package \texttt{volesti} \cite{volesti}, can be used; this tool does not rely on triangulation but instead employs random-walk-based methods to provide uniform sampling from a given polytope. Examples of applications in both low and high dimensions can be found in \cite{fontana2023exchangeable}. In high dimensions, another possible approach consists of using MCMC-type methods. For instance, Markov bases from Algebraic Statistics have been successfully employed in related problems that require navigating subsets of the probability simplex \cite{PRR2021}. A complete study of these techniques, however, is beyond the scope of the present paper.

\section{Illustrative examples and applications}\label{sec:examples}

In this section, we illustrate how the proposed framework can be applied to real and synthetic contingency tables. 
We begin by analyzing a four-way binary table, showing how the feasible set of distributions with fixed pairwise dependence and uniform margins can be characterized through its extreme rays. We then explore how the presence of observed zeros affects the geometry of the feasible region, potentially restricting or altering the dependence structure. 
Finally, we present a real-data application in the context of rater agreement, demonstrating how extreme solutions can help interpret higher-order dependence and guide model selection.

\subsection{Four-way binary table: geometric structure and extreme rays}\label{sec:example4way}

Having previously illustrated the case $d=3$, we now consider a four-way binary table to demonstrate how the proposed approach extends to higher dimensions. This example, adapted from \cite{goodman:71}, allows us to explicitly construct the matrix $H_4$, determine the corresponding polyhedral cone, and compute its extreme rays. 
Unlike the three-way case discussed in Section~\ref{sec:framework}, the four dimensional case reveals a more complex geometry, with a substantially larger number of extreme rays and more complex patterns of higher-order dependence. 

The original table presented in \cite{goodman:71} is of dimension $3 \times 2^3$. We adapted it to our framework by collapsing two levels of the first variable. The resulting table, reported in Table \ref{tab:water}, summarizes the cross-classification of a sample of $1,008$ consumers according to four variables: the softness of the laundry water used (Soft, Medium$+$Hard); the previous use of a detergent M; the temperature (T) of the laundry water used (High, Low); the preference for Detergent X over M. 

\begin{table}[t]
\caption{Four-way contingency table for the example in Sect.~\ref{sec:example4way}. Data adapted from \cite{goodman:71}.}\label{tab:water}%
\begin{tabular}{llcccc} \toprule
 & & \multicolumn{2}{c}{Previous user of M} & \multicolumn{2}{c}{Previous non-user of M} \\
Water softness & Brand preference & High T & Low T & High T & Low T \\
\midrule

{Soft}   & X & 19 & 57 & 29 & 63 \\
  & M & 29 & 49 & 27 & 53 \\
\midrule
{Medium$+$Hard}   & X & 47 & 84 & 75 & 134 \\
  & M & 90 & 107 & 53 & 92 \\
\bottomrule
\end{tabular}
\end{table}

The pmf $p_0$ corresponding to the four-way contingency table of Table \ref{tab:water} is computed by normalization. 
	

The marginal odds ratios corresponding to $p_0$ are reported in Table \ref{tab:waterorm}. Analogously to what we did for three-way tables, given $\omega_M^{ij}$, the equation Eq.~\eqref{eq:ormp11} can be used to determine the corresponding second-order moments $\mu_{ij}$, $i,j \in\{1,2,3,4\}, i<j$, that will provide the conditions to get pmfs with uniform margins and the same marginal odd-ratios of $p_0$. These second-order moments are reported in Table \ref{tab:waterorm}. Building the corresponding $H_4$ matrix and using \texttt{4ti2}, \cite{4ti2}, we obtain the extreme rays that, after normalization, provide the extreme pmfs of the polytope that contains all the pmfs with uniform margins and the same marginal odds ratio of the original table $p_0$. We find $n_d=96$ extreme pmfs.

\begin{table}[h]
	\centering
    \caption{Marginal odds ratios of $p_0$ and second-order moments of the pmfs with uniform margins and the same marginal odds ratios of $p_0$ for Table \ref{tab:water}, all values rounded to 3 decimal places.}
	\label{tab:waterorm}
		\begin{tabular}{ccrr}  \toprule
        $i$ & $j$ & $\omega_M^{ij}$ & $\mu_{ij}$ \\ 
        \midrule
1 &	2 &	1.070 &	0.254 \\
2 &	3 &	0.563 &	0.214\\
1 &	3 &	0.966 &	0.248\\
3 &	4 &	1.158 &	0.259\\
2 &	4 &	0.761 &	0.233\\
1 &	4 &	0.737 &	0.231\\
\bottomrule
\end{tabular}
	
\end{table}

By forming a convex combination of these extreme pmfs, we obtain a rich variety of admissible joint distributions that all share the same pairwise dependence structure, thereby allowing for a more flexible selection of the most suitable model for further analysis.

\subsection{Empirical application on a rater agreement problem}\label{sec:case_study}

We now illustrate how the proposed framework can be used in practice using real agreement data involving three raters. In such settings, the interest typically lies not only in pairwise agreement, but also in detecting higher-order agreement patterns that cannot be captured by pairwise measures alone. 
The data, presented in Table \ref{tab:agree}, are treated as a binary table in \cite{vanbelle:23}, obtained by collapsing a three-category table originally reported in \cite{martinandres:20}. 
In this study on persuasive communication, three raters classified the cognitive responses of $164$ patients into two categories: “positive or neutral” (level $2$) and “negative” (level $1$).
  
\begin{table}[h!]
\caption{Table from rater agreement problem. Data from \cite{vanbelle:23}.}\label{tab:agree}%
\begin{tabular}{@{}cccc@{}}  \toprule
 &  & \multicolumn{2}{c}{R3} \\
R1 & R2 & 1 & 2 \\
\midrule
1 & 1 & $113$ & $5$ \\
   & 2 & $5$ & $7$ \\
\midrule
2 & 1 & $4$ & $3$\\
   & 2 & $3$ & $24$ \\
\botrule
\end{tabular}
\end{table}

As already discussed, standard procedures, such as MIPFP or log-linear models, return a single table and often implicitly force higher-order interactions to vanish, which may lead to misleading interpretations when genuine higher-order dependence is present.
For example, the three-way interaction in  Table~\ref{tab:agree} can be described by the three-dimensional odds ratio, whose observed value is
\[
\widehat{\omega}^{123} = \frac {n_{111}n_{122}n_{212}n_{221}} {n_{222}n_{211}n_{121}n_{112}} = 2.96625.
\]
This value suggests a strong pattern of dependence beyond the two-way correlations, thus making it not reasonable to use, for example, the MIPFP-based algorithm (which takes the high-order odds ratios equal to $1$) for simulation purposes in frameworks of this kind. 
By instead characterizing all admissible pmfs with uniform margins and fixed pairwise dependence, our framework allows us to investigate how higher-order interactions can vary while keeping pairwise agreement unchanged. 
This provides a richer and more transparent view of agreement structure and helps determine whether observed dependence can be explained by pairwise associations alone, or whether a three-way agreement effect is present.
The pmf $p_0$ corresponding to the three-way contingency table of Table \ref{tab:agree} is reported in Table \ref{tab:raterpmf}.

\begin{table}[t]
	\centering 
    	\caption{$p_0$ and its marginal odds ratios for Table \ref{tab:agree}, all values rounded to 3 decimal places.}
	\label{tab:raterpmf}
		\begin{tabular}{lrrrrrrrrrrr} \toprule
pmf & $000$ & $001$ & $010$ & $011$ & $100$ & $101$ & $110$ & $111$ & $\omega_M^{12}$ &
$\omega_M^{13}$ & $\omega_M^{23}$\\
\midrule
$p_0$ & $0.689$ & $0.03$ & $0.03$ & $0.043$ & $0.024$ & $0.018$ & $0.018$ & $0.146$ & $37.929$ & $37.929$ & $56.672$ \\  \bottomrule
\end{tabular}

\end{table}
\begin{table}[t!]
	\centering
    	\caption{Extreme pmfs $r_1$ and $r_2$ for $p_0$ in Table \ref{tab:raterpmf}, all values rounded to 3 decimal places.}
	\label{tab:erexrealcase} 
		\begin{tabular}{lrrrrrrrr} \toprule
pmf & $000$ & $001$ & $010$ & $011$ & $100$ & $101$ & $110$ & $111$ \\
\midrule
$r_1$ & $0.372$ & $0.059$ & $0.059$ & $0.011$ & $0.07$ & $0$ & $0$ & $0.43$ \\ 
$r_2$ & $0.43$ & $0$ & $0$ & $0.07$ & $0.011$ & $0.059$ & $0.059$ & $0.372$  \\  
\bottomrule
\end{tabular}

\end{table}

Following the steps already discussed in the previous sections, we can derive the extreme pmfs, which are reported in Table \ref{tab:erexrealcase}.
The values of the parameters associated with $r_1$ and $r_2$ for both parametrizations
are reported in Table \ref{tab:betarealcase}. The values of $r_{i,\alpha}$ have been set to $r_{i,\alpha}+10^{-8}$ to avoid numerical problems in computing logarithms.

%
%
%

\begin{table}[h]
	\centering 
    	\caption{Parameters of the extreme pmfs $r_1$ and $r_2$ in Table \ref{tab:erexrealcase}, all values rounded to 2 decimal places.}
	\label{tab:betarealcase}
		\begin{tabular}{llrrrrrrrr} \toprule
parametrization & pmf & $\lambda$ & $\lambda_1^1$ & $\lambda_1^2$ & $\lambda_1^3$ & $\lambda_{11}^{11}$ & $\lambda_{11}^{13}$ & $\lambda_{11}^{23}$ & $\lambda_{111}^{111}$ \\
\midrule
corner & $r_1$ & $-0.99$ & $-1.67$  & $-1.85$ & $-1.85$ & $-13.91$ & $-13.91$ & $0.19$ &   $33.14$ \\ 
corner & $r_2$ & $-0.84$ & $-3.65$ & $-17.58$ & $-17.58$  & $19.23$  & $19.23$ & $33.34$  & $-33.14$ \\ \midrule
zero-mean & $r_1$ & $-6.44$ & $-3.65$ & $-0.21$ & $-0.21$ & $0.66$ & $0.66$ & $4.19$  & $4.14$ \\ 
zero-mean & $r_2$ & $-6.44$ & $3.65$ & $0.21$ & $0.21$ & $0.66$ & $0.66$ & $4.19$  & $-4.14$  \\  \bottomrule

\end{tabular}

\end{table}

\section{Discussion and future work}\label{sec:discussion}
This paper introduces a geometric framework for characterizing the full set of joint distributions of $d$ binary variables with uniform margins and fixed pair-wise dependence, expressed via correlations or marginal odds ratios. 
Rather than selecting a single representative distributions, as commonly done in existing approaches such as MIPFP or log-linear model based sampling, we describe the entire admissible space of solutions. 
We show that this feasible set forms a polyhedral cone, study its extreme rays and prove some symmetry properties. By identifying these rays, we gain insight into how higher-order interactions may vary while keeping pair-wise information fixed.
This perspective allows for a more conscious and flexible choice of the dependence structure when generating tables.

We now discuss some limitations of the approach presented in this work. Although assuming uniform margins offers both mathematical and interpretative advantages, this choice may also introduce restrictions, particularly when structural or sampling zeros are present in the data. Zero entries can substantially reduce the feasible set of distributions, and in some cases make it impossible to construct a table with uniform margins and prescribed marginal odds ratios while preserving the original support. In such cases, understanding how zero-patterns interact with uniform margins becomes essential. This aspect has been studied in \cite{fontana2025IES, fontana2026ijar}, where conditions for compatibility are provided. These results are fundamental to our framework, as they can be adjusted to allow us to identify when the feasible set is nonempty and to interpret the implications of zero-patterns on the structure of admissible distributions.

From a computational perspective, several challenges remain open. As discussed in Section \ref{sec:framework}, although there are established methods for sampling points from convex polytopes, their suitability and efficiency in the context of our specific geometric structure have not yet been fully explored.
Assessing and adapting these sampling strategies to our framework is a promising direction for future work and will be addressed in a follow-up study.
Besides, the framework developed here focuses on binary tables. A natural direction for future work is to extend the framework to multilevel categorical variables, which frequently arise in applications. Such a multilevel settings introduce larger and more complex feasible regions, resulting in a more challenging and interesting geometric characterization. We believe this will further broaden the applicability of our approach.

\section*{Declarations}

The authors declare no competing interests.



\end{document}